\documentclass[a4paper,fleqn]{cas-sc}

\usepackage[numbers]{natbib}
\usepackage{tikz}
\usepackage{algorithm}
\usepackage{algpseudocode}
\usepackage{amsmath}
\usepackage{amsfonts}

\algrenewcommand\algorithmicrequire{\textbf{Input:}}
\algrenewcommand\algorithmicensure{\textbf{Output:}}

\let\set\mathbb
\def\<#1>{\langle#1\rangle}

\def\eatspace#1{#1}
\def\step#1#2{%
  \par\kern1pt\dimen144=#2em\advance\dimen144by1.67em
  \hangindent=\dimen144\hangafter=1
  \leavevmode\rlap{\small#1}\kern\dimen144\relax\eatspace}

\newcommand{\N}{\mathbb{N}}
\newcommand{\R}{\mathbb{R}}

\renewcommand{\O}{\operatorname{O}}
\renewcommand{\Gamma}{\varGamma}
\newcommand{\A}{\mathbf{A}}
\newcommand{\B}{\mathbf{B}}
\def\C{\mathbf{C}}

\newcommand{\Tens}[1]{\langle #1 \rangle}

\newtheorem{theorem}{Theorem}
\newtheorem{proposition}[theorem]{Proposition}
\newtheorem{lemma}[theorem]{Lemma}
\newdefinition{definition}[theorem]{Definition}

\newtheorem{corollary}[theorem]{Corollary}
\newdefinition{question}{Question}
\newdefinition{remark}{Remark}
\newproof{proof}{Proof}

\shorttitle{Exploiting Structure in Tensor Decompositions for Matrix Multiplication}
\shortauthors{M. Kauers et al.}

\begin{document}

\title[mode=title]{Exploiting the Structure in Tensor Decompositions for Matrix Multiplication}

\tnotemark[1,2,3]
\tnotetext[1]{M. Kauers was supported by the Austrian FWF grants 10.55776/PAT8258123,
	10.55776/I6130, and 10.55776/PAT9952223}
\tnotetext[2]{J. Moosbauer was supported by the Austrian FWF grant 10.55776/PAT8258123}
\tnotetext[3]{I. Wood was supported by the Austrian FWF grant 10.55776/PAT8258123}

\author[1]{Manuel Kauers}[orcid=0000-0001-8641-6661]
\ead{manuel.kauers@jku.at}

\author[1]{Jakob Moosbauer}[orcid=0000-0002-0634-4854]
\ead{jakob.moosbauer@jku.at}

\author[1]{Isaac Wood}
\ead{isaac.wood@jku.at}

\affiliation[1]{
	organization={Institute for Algebra, Johannes Kepler University},
	city={Linz},
	postcode={A4040},
	country={Austria}
}

\begin{abstract}
	We present a new algorithm for fast matrix multiplication using tensor
	decompositions which have special features.
	Thanks to these features we
	obtain exponents lower than what the rank of the tensor decomposition
	suggests.
	In particular for $6\times 6$ matrix multiplication we reduce the
	exponent of the recent algorithm by Moosbauer and Poole from $2.8075$ to
	$2.8019$, while retaining a reasonable leading coefficient.
\end{abstract}

\begin{keywords}
	Bilinear complexity \sep Strassen's algorithm \sep Tensor rank
\end{keywords}

\maketitle

\section{Introduction}\label{sec:introduction}

Since Strassen's seminal discovery that two $n\times n$ matrices can be multiplied in $\O(n^{2.81})$ arithmetic operations~\cite{STRASSEN1969}, there has been a race to reduce the upper bound on $\omega$, the exponent of matrix multiplication.
Following Strassen's breakthrough, Bini et al.~\cite{BINI1979} introduced the concept of approximate algorithms and border rank, allowing for a small improvement on the upper bound, finding $\omega \leq 2.78$.
This line of research continued to Sch\"onhage's Asymptotic Sum Inequality (ASI)~\cite{SCHONHAGE1981}, allowing for algorithms that compute disjoint matrix multiplications simultaneously, giving a bound of $\omega \leq 2.52$. Later improvements, based on Strassen's laser method~\cite{STRASSEN1986} and the Coppersmith-Winograd approach~\cite{COPPERSMITH1990}, further reduced the bound on $\omega$ to the currently best upper bound of $\omega < 2.371339$ by Alman et al.~\cite{ALMAN2024}.

While there has been a lot of work done to reduce the exponent of matrix multiplication, approximate algorithms (like those used in all the fastest algorithms since 1979) in general either require unreasonable precision or have such large leading coefficients that they are
impractical. Most current research into practical algorithms is based on exact algorithms, either trying to reduce the number of multiplications (using, e.g., flip graphs \cite{KAUERS2022,KAUERS2023,KAUERS2025,MOOSBAUER2023,MOOSBAUER2025,PERMINOV2025,WOOD2025}, numerical optimization~\cite{SMIRNOV2013, KAPORIN2024, NOVIKOV2025}, or reinforcement learning~\cite{FAWZI2022}) or reducing the number of additions (Winograd's variant of Strassen's algorithm, published in~\cite{PROBERT1976}, alternative basis algorithms~\cite{KARSTADT2017, BENIAMINI2019, SCHWARTZ2024}, flip graph search~\cite{PERMINOV2025}, or common subexpression elimination~\cite{MARTENSSON2025}).
Sch\"onhage's 1981 paper~\cite{SCHONHAGE1981} uses that an algorithm by Pan~\cite{PAN1984} has a special property. After 3 iterations, 8 of the recursive calls form a $2\times 2$ matrix multiplication, which can then be computed using Strassen's algorithm, thereby saving one multiplication.  Schwartz and Zwecher~\cite{SCHWARTZ2025} recently used a similar idea to improve some of Pan's trilinear aggregation algorithms.  The concept also aligns with Romani's generalization of the ASI~\cite{ROMANI1981}, which allows for a direct sum of asymmetric matrix multiplication tensors on the right-hand side.  Schwartz and Zwecher give an explicit algorithm with a $44 \times 44$ base
case, that achieves the best exponent among exact algorithms with base case smaller than $1000$, while Sch\"onhage's and Romani's results are purely theoretical bounds on the exponent, with no explicit algorithm provided.
In this paper we propose a new recursive matrix multiplication algorithm that uses the same idea.
If some recursive calls share one of the inputs or have an output that is used in multiple positions, then they are treated as a single matrix multiplication of larger size.
Applying this technique repeatedly improves the exponent of a matrix multiplication algorithm without reducing the number of multiplications in the base case.
In particular, we improve the recent $6\times 6$ matrix multiplication algorithm by Moosbauer and Poole~\cite{MOOSBAUER2025} from an exponent of $2.8075$ to $2.8019$, compared to the $2.8073$ exponent of Strassen's algorithm.
Our algorithm outperforms the standard algorithm for $n \geq 1000$ in terms of total operation count.
This does not necessarily mean that in an implementation one would already observe a speedup for such matrix sizes, but it clearly puts it in the realm of practical algorithms by Pan's definition~\cite{PAN1984}, who generously drew the line at $n=10^{20}$.
We ran a search for such algorithms,  first using a flip graph search to find algorithms using a minimal number of multiplications while retaining the special structure needed to apply our new recursive method.
Then we optimize the number of additions using DeGroote actions.
These steps produce algorithms that are valid only modulo~$2$, thus we apply Hensel lifting to lift the algorithms to work over the integers, and finally we reduce the number of additions further by M\aa rtensson and Wagner's common subexpression elimination method~\cite{MARTENSSON2025}.
We find improvements in the exponent for several small base cases including $3\times 3$ matrix multiplication, where we reduce the exponent from Laderman's $2.854$~\cite{LADERMAN1976} to $2.836$ (for $K=\set Z_2$) or $2.843$ (for any field).
\section{Background} \label{sec:background}

Let $K$ be a field and $R$ be a $K$-algebra, and let $\A\in R^{n \times m},\B\in R^{m\times p}$.
Strassen's algorithm~\cite{STRASSEN1969} uses a divide and conquer strategy to multiply matrices faster than the standard algorithm.
First, it reduces the multiplication of two $n\times n$ matrices to multiplying two $2\times 2$ matrices whose entries are $\frac{n}{2} \times \frac{n}{2}$ matrices using block matrix multiplication.
Then it computes this $2\times 2$ block matrix product using only $7$ multiplications of the subblocks by linearly combining them in a clever way.
Using this strategy recursively gives an algorithm to multiply $n\times n$ matrices using $\O(n^{\log_2(7)})$ operations in $R$.
The exponent only depends on the number of multiplications, since the cost of the recursive call dominates the cost of the linear combinations.
Usually, these algorithms are written using the language of tensors.
\begin{definition}
	Let $n,m,p\in \N$ and $A = R^{n\times m},B = R^{m\times p}, C = R^{p\times n}$.
	Denote by $\{a_{ij} : 1 \leq i \leq n, 1 \leq j \leq m\}$ the standard basis of $A^*$, $\{b_{jk} : 1 \leq j \leq m, 1 \leq k \leq p\}$ the standard basis of $B^*$ and $\{c_{ki} : 1 \leq k \leq p, 1 \leq i \leq n\}$ the standard basis of $C$.  So $a_{ij}$ is the linear functional that takes a matrix in $A$ and returns its $(i,j)$-th entry, and similarly for $b_{jk}$.
	We define the matrix multiplication tensor to be
	\[
		\Tens{n,m,p} = \sum_{i,j,k=1}^{n,m,p} a_{ij} \otimes b_{jk} \otimes c_{ki} \in A^* \otimes B^* \otimes C .
	\]
\end{definition}

This tensor encodes matrix multiplication in the following sense:  denote by $\operatorname{Bil}(A,B;C)$ the space of all bilinear maps from $A\times B$ to $C$.
Then we can define an isomorphism $\Phi \colon A^* \otimes B^* \otimes C \to \operatorname{Bil}(A,B;C)$ by
\[
	\Phi(a_{i_2j_1}\otimes b_{j_2k_1} \otimes c_{k_2i_1}) ( \A, \B) = a_{i_2j_1}(\A) b_{j_2k_1}(\B)\cdot \C_{i_1k_2} = \A_{i_2j_1} \B_{j_2k_1} \C_{i_1k_2}
\]
where $\A \in A,\B \in B$ are matrices, $\A_{ij},\B_{ij}$ are the $(i,j)$th entries of $\A$ and~$\B$, respectively,
and $\Gamma_{ij}\in C$ is the matrix with all zeros except for a $1$ at the $(i,j)$-th entry.
Note the $c_{ij}$ corresponds to $\C_{ji}$ in the product matrix;
this is a standard convention to keep the cyclic nature of matrix multiplication.
\begin{definition}
	We call a tensor $T\in A^* \otimes B^* \otimes C$ a restriction of $T' \in A'^* \otimes B'^* \otimes C'$, denoted by $T \leq T'$, if there are homomorphisms $\phi_A \colon A \to A',\phi_B \colon B \to B',\phi_C \colon C' \to C$ with $T = (\phi_A^* \otimes \phi_B^* \otimes \phi_C) T'$.
	If we need to refer to the specific homomorphism $\phi = (\phi_A^*\otimes \phi_B^* \otimes \phi_C)$, we write $T \leq_\phi T'$.

	If $\phi_A,\phi_B,\phi_C$ are in fact isomorphisms, we call $T$ and $T'$ isomorphic and write $T\cong T'$.

	We write the unit tensor of size $k$ as $\Tens{k} = \sum_{i=1}^k e_i\otimes e_i \otimes e_i \in R^k \otimes R^k \otimes R^k$ where $\{e_1,\dots,e_k\}$ is the standard basis for $R^k$.
\end{definition}

If a tensor $T$ is a restriction of a tensor $T'$ this means that any algorithm to compute $\Phi(T')$ can be transformed into one to compute $\Phi(T)$ by applying the linear maps $\phi_A,\phi_B,\phi_C$ to the inputs and output.
In particular, if we have a restriction of the form $\Tens{n,m,p} \leq \Tens{r}$, then we have an algorithm to multiply $n\times m$ and $m\times p$ matrices using $r$ multiplications in $R$.
A restriction of the form $T \leq \Tens{r}$ gives rise to a decomposition of $T$ into $r$ summands of the form
\[
	T = \sum_{i=1}^r a_i\otimes b_i\otimes c_i.
\]

We then call $r$ the \emph{rank} of this decomposition and a tensor that can be written in the form $a \otimes b \otimes c$ a \emph{rank-one} tensor.
The rank of a tensor $T$ is defined as the smallest integer $r$ such that $T$ has a decomposition of rank $r$.
Together with Strassen's recursive block matrix multiplication strategy, this gives the following fundamental result.
\begin{theorem}\label{thm:basicexponent} 
	If $\Tens{n,m,p}\leq \Tens{r}$ then there is an algorithm to multiply $N\times N$ matrices in $\O(N^{3\log_{nmp}(r)})$ operations in $R$.
\end{theorem}

Using the decomposition Strassen found, we get $\Tens{2,2,2} \leq \Tens{7}$ and hence $\omega \leq \log_2(7)$.
Further analysis allows to compute the leading coefficient of this algorithm as well, as we will see next.
While the exponent is given by the rank of a decomposition, i.e. the number of multiplications, the leading coefficient depends on the number of additions and scalar multiplications needed in the restriction.
We will denote this number by $A(\phi)$ if the restriction is given by~$\phi$.
The homomorphism for a restriction is in general not unique, but in our case every restriction is witnessed by an explicit homomorphism, so we just write $A$ when the homomorphism is clear from the context.
Using this definition, Strassen's original algorithm has $A=18$. Using this we can compute the leading coefficient of the algorithm, which is done in general in Theorem \ref{thm:basiccost}.
Although this result is well-known, we present its proof in detail since as it establishes the structure for the more complex proof in Section \ref{sec:newstuff}.
\begin{theorem} \label{thm:basiccost}
	For a restriction $\Tens{n,n,n} \leq_\phi \Tens{r}$ with $ \omega_0 = \log_n(r) < 3$, the algorithm from Theorem \ref{thm:basicexponent} needs at most
	\[
		\left( 2(n-1)^{3-\omega_0} + \frac{r(2^{\omega_0}-1)+4A(\phi)}{r-n^2}(n-1)^{2-\omega_0} \right) N^{\omega_0} + \O(N^2)
	\]
	operations in $K$ to multiply $N \times N$ matrices.
\end{theorem}
\begin{proof}
	We denote by $T(N)$ the total number of operations the algorithm needs to multiply $N \times N$ matrices.
	For $1\leq N < n$ we use the standard algorithm, so $T(N) = 2N^3 - N^2$.
	For $N \geq n$, we compute $r$ recursive multiplications of size $\lceil N/n \rceil \times \lceil N/n \rceil$, and need $A \lceil N/n \rceil^2$ operations to form the linear combinations.
	The ceiling accounts for zero padding to the next multiple of $n$.
	Hence, we have the recurrence relation
	\[
		T(N) = r T(\lceil N/n \rceil) + A \lceil N/n \rceil^2
	\]
	for all $N\geq n$.
	We will show by induction on $N$ that
	\[
		T(N) \leq L N^{\omega_0} - d N^2
	\]
	for $L = 2(n-1)^{3-\omega_0} + \frac{r(2^{\omega_0}-1)+4A}{r-n^2}(n-1)^{2-\omega_0}$ and $d$ to be determined soon.
	Suppose the statement holds for all $N < k$. Then we have

	\begin{align*}
		T(k) & = r T(\lceil k/n \rceil) + A \lceil k/n \rceil^2                                        \\
		     & \leq r (L \lceil k/n \rceil^{\omega_0} - d \lceil k/n \rceil^2) + A \lceil k/n \rceil^2 \\
		     & \leq r (L ( k/n + 1 )^{\omega_0} - d
		(k/n)^2) + A (k/n+1)^2                                                                         \\
		     & \leq r (L (k/n)^{\omega_0} + c(k/n)^{\omega_0-1} - d (k/n)^2) + A (k/n+1)^2             \\
		     & \leq L k^{\omega_0} + \frac{rc}{n^2}k^2 - \frac{rd}{n^2} k^2 + \frac{4A}{n^2} k^2,
	\end{align*}
	where $c = 2^{\omega_0}-1$ is chosen such that $(x/n+1)^{\omega_0} \leq (x/n)^{\omega_0} + c (x/n)^{\omega_0-1}$ for all $x \geq n$.
	Choosing $d = \frac{rc+4A}{r-n^2}$ completes the induction step.

	It remains to show that for all $N < n$ we have $T(N) \leq L N^{\omega_0} - d N^2$.
	Since $T(N) = 2N^3 - N^2$ for $N < n$, it suffices to show that $L \geq 2N^{3-\omega_0} + (d-1) N^{2-\omega_0}$ for all $1 \leq N < n$, which is the case.
	\qed\end{proof}

In the case of Strassen's original algorithm, which used a restriction $\Tens{2,2,2} \leq_\phi \Tens{7}$ with $A(\phi) = 18$, we find that this theorem gives us a bound on the leading coefficient $L \leq 40$.
Usually the leading coefficient bounds are computed in the idealized case where every recursive call is assumed to be done exactly.
In this case, the leading coefficient of Strassen's algorithm is~$7$, the general formula being $L \leq \frac{A}{r-n^2}+1$~\cite{KARSTADT2017}. While
Theorem~\ref{thm:basicexponent} gave a substantial improvement to the naive approach, all improvements since 1981 to the upper bound on $\omega$ have come from the following stronger theorem due to Sch\"onhage.

\begin{theorem}[Asymptotic Sum Inequality, ASI] \label{thm:ASI}
	If
	\[
		\bigoplus_{i=1}^k \Tens{n_i,m_i,p_i} \leq \Tens{r},
	\]
	and $\omega_0$ is such that
	\[
		\sum_{i=1}^k (n_im_ip_i)^\frac{\omega_0}{3} = r,
	\]
	then for every $\epsilon>0$, there exists an algorithm to multiply $n\times n$ matrices using $\O(n^{\omega_0+\epsilon})$ operations in~$R$.
\end{theorem}

Our result is closely related to the following generalization of the ASI, which was first stated by Romani~\cite{ROMANI1981}.

\begin{theorem}
	If
	\[
		\bigoplus_{i=1}^k \Tens{n_i,m_i,p_i} \leq \bigoplus_{i=1}^{q} \Tens{n_i',m_i',p_i'},
	\]
	and $\omega_0$ is such that
	\[
		\left(\sum_{i=1}^k (n_im_ip_i)^\frac{\omega_0}{3} \right)^3 = \left(\sum_{i=1}^{q} n_i'^{\omega_0-2}m_i'p_i'\right)\left(\sum_{i=1}^{q} n_i'm_i'^{\omega_0-2}p_i'\right)\left(\sum_{i=1}^{q} n_i'm_i'p_i'^{\omega_0-2}\right),
	\]
	then for every $\epsilon>0$, there exists an algorithm to multiply $n\times n$ matrices using $\O(n^{\omega_0+\epsilon})$ operations in~$R$.
\end{theorem}

While the ASI has proved very useful for finding ever tighter bounds on the exponent of matrix multiplication, it has, to our knowledge, never been used for finding fast matrix multiplication algorithms outside the purely asymptotic regime.
The theorem is frequently used in the context of approximate algorithms, requiring either an unreasonable level of precision for most applications or an extremely large leading coefficient, as well as very large base cases, rendering them functionally useless in the context of practical algorithms.
In this work we do not aim to improve the exponent of matrix multiplication in general, but to find algorithms which could be used in practical computations, while achieving a better complexity.
This does not mean that the algorithms presented here are practical in the sense of being faster than currently used algorithms for matrix sizes that appear in practical computations, but at least they could reasonably be implemented and used for matrix sizes that fit into the memory of a large computer.
We will use direct sums and Kronecker products of tensors in the following part.
The important property is that they are compatible with the notion of restriction, so we have
\[
	T_1 \leq T_1', T_2 \leq T_2' \implies T_1 \oplus T_2 \leq T_1' \oplus T_2',\quad T_1 \otimes T_2 \leq T_1' \otimes T_2'.
\]
Readers less familiar with these concepts can view these operations as follows: The direct sum $\oplus$ corresponds to computing two disjoint matrix multiplications simultaneously, while the Kronecker product $\otimes$ corresponds to the inner and outer part of a block matrix multiplication.
In particular, we have that
\[
	\Tens{n_1,m_1,p_1} \otimes \Tens{n_2,m_2,p_2} \cong \Tens{n_1 n_2, m_1 m_2, p_1 p_2}.
\]
Since we will frequently encounter direct sums  and Kronecker products of matrix multiplication tensors of the same format, we use the notation
\[
	k \odot \Tens{n,m,p} = \bigoplus_{i=1}^k \Tens{n,m,p}, \quad \Tens{n,m,p}^{\otimes k} = \bigotimes_{i=1}^k \Tens{n,m,p}.
\]

We will also use the well-known fact that one can permute the dimensions of matrix multiplication tensors.
Specifically, we get from $(\A\B)^T = \A^T\B^T$ that we transform any algorithm from a restriction of $\Tens{n,m,p}$ into one for $\Tens{p,m,n}$ and from the cyclic symmetry of the matrix multiplication tensor it follows that we can transform it to $\Tens{m,p,n}$ as well.
\section{Divide less, conquer more}\label{sec:newstuff}

While Sch\"onhage's ASI is an existence result, we present an explicit algorithm with the asymptotic complexity given by the generalized ASI.  In contrast to the classical ASI we do not allow a direct sum of matrix multiplication tensors on the left side, but only on the right side.  This means that we decompose a matrix multiplication tensor into a sum of smaller matrix multiplications, not just rank-one tensors.  This allows to divide a matrix multiplication into larger subblocks, thereby performing more recursive steps (and hence more savings compared to the standard algorithm).

The concept
of achieving additional savings this way originates from Sch\"onhage's 1981 paper~\cite{SCHONHAGE1981}, in which it is shown to find improvements to Pan's use of the ASI~\cite{PAN1984} by taking tensor powers and then applying Strassen's algorithm.
Let us illustrate the idea using the decomposition of $\Tens{6,6,6}$ found by Moosbauer and Poole~\cite{MOOSBAUER2025}.
They show that $\Tens{6,6,6}$ can be decomposed into a sum of $153$ rank-one tensors, analyzing the structure of this decomposition, we can find the restriction
$\Tens{6,6,6} \leq 137 \odot \Tens{1} \oplus 8 \odot \Tens{1,1,2}$.
Applying a cyclic permuation $\sigma$ to this restriction gives $\Tens{6,6,6} \leq 137 \odot \Tens{1} \oplus 8 \odot \Tens{2,1,1}$, and $\Tens{6,6,6} \leq 137 \odot \Tens{1} \oplus 8 \odot \Tens{1,2,1}$.
If we then consider the Kronecker product of these three restrictions we get
\begin{align*}
	\Tens{6,6,6}^{\otimes3} & \leq    137^3 \odot \Tens{1} \oplus 8 \cdot 137^2 \odot (\Tens{1,1,2} \oplus \Tens{2,1,1} \oplus \Tens{1,2,1})              \\
	                        & \kern-1pt \oplus  8^2 \cdot 137 \odot (\Tens{1,2,2} \oplus \Tens{2,1,2} \oplus \Tens{2,2,1}) \oplus 8^3 \odot \Tens{2,2,2}.
\end{align*}
While we have to compute $\Tens{1,1,2},\Tens{1,2,2}$ and their cyclic permutations by the standard algorithm, using $2$ and $4$ multiplications respectively, we can compute $\Tens{2,2,2}$ using Strassen's algorithm with $7$ multiplications.
In the language of tensor restrictions we use $\Tens{1,1,2} \leq \Tens{2}$, $\Tens{1,2,2} \leq \Tens{4}$ and $\Tens{2,2,2} \leq \Tens{7}$, to conclude that $\Tens{216,216,216} \leq \Tens{3581065}$.
Applying Theorem \ref{thm:basicexponent} to this restriction gives an algorithm with exponent $\omega_0 = 2.80751$, which is only slightly smaller than the original exponent $2.80754$.
We can do even better by taking higher tensor powers.

But this is not all.
The key observation is that instead of using Strassen's decomposition to save additional multiplications, we can use any fast matrix multiplication algorithm to further generalize the matrix multiplication tensors we restricted from, in particular we can use the same strategy recursively.
In this example we would get an algorithm with exponent $2.805065$.
We will show that we can use this process in a recursive matrix multiplication algorithm that is very similar to the standard Strassen-like algorithms obtained from tensor decompositions that has the complexity as given by the generalized ASI while retaining a reasonable leading coefficient.
In Algorithm \ref{alg:MatrixMults} we write $\mathcal{L}_{1}^{(i,j)},\mathcal{L}_{2}^{(i,j)}, \mathcal{L}_3^{(k,l)}$ for the functions that compute the linear combinations of sub-blocks as given by the tensor restriction.
Unlike in Strassen's algorithm, some of these block matrices are recombined into larger matrices of size $n_i \times m_i$ and $m_i \times p_i$ respectively.
\begin{algorithm}
	\caption{Matrix Mult for a given restriction $\Tens{n,m,p} \leq \bigoplus_{i=1}^q s_i\odot \Tens{n_i,m_i,p_i}$ and a given threshold $N_0$}
	\label{alg:MatrixMults}
	\begin{algorithmic}[1]
		\Require An $N \times M$ matrix $\A$ and an $M \times P$ matrix $\B$.
		\Ensure The matrix product $\C = \A\B$
		\Function{MatrixMults}{$\mathbf{A},\mathbf{B}$}
		\If{$\min(N,M,P) \leq N_0$}
		\State \Return $\A\B$ computed by the standard algorithm
		\EndIf
		\If{$n\nmid N \lor m\nmid M \lor p\nmid P$}
		\State $N \gets \lceil N/n \rceil n$
		\State $M \gets \lceil M/m \rceil m$
		\State $P \gets \lceil P/p \rceil p$
		\State pad $\A$ and $\B$ with zeros to dimensions $N\times M$ and $M\times P$
		\EndIf
		\For{$i\gets 1$ to $q$}
		\For{$j \gets 1$ to $s_i$}
		\State $X \gets \mathcal{L}_{1}^{(i,j)}(\A)$
		\State $Y \gets \mathcal{L}_{2}^{(i,j)}(\B)$
		\State $Z_{i,j} \gets \Call{MatrixMult}{X,Y,Nn_i/n,Nm_i/m,Pp_i/p}$
		\EndFor
		\EndFor
		\For{$k \gets 1$ to $N/n$}
		\For{$l \gets 1$ to $P/p$}
		\State $\C_{k,l} \gets \mathcal{L}_3^{(k,l)}(\mathbf{Z})$
		\EndFor
		\EndFor
		\State \Return $\C$
		\EndFunction
	\end{algorithmic}
\end{algorithm}

Algorithm \ref{alg:MatrixMults} correctly computes the matrix product $\mathbf{C}$, since it computes all the same matrix products as the standard
Strassen-like algorithm obtained from the tensor decomposition. However, our algorithm  computes some of them simultaneously as a larger matrix.
If in Algorithm~\ref{alg:MatrixMults} some of the recursive calls are rectangular matrix multiplications (as will be the case in all our results), then for some recursive paths the number of recursions is limited to $\min(\log_{\frac{n_i}{n}}(N), \log_{\frac{m_i}{m}}(M), \log_{\frac{p_i}{p}}(P))$, before we encounter a subproblem of shape $\Tens{N',M',1}$, or a permutation thereof, which forces us to switch to the standard algorithm.
However, we do not need to use the same restriction in each recursive call.
Similar to the example discussed above we can alternate between three cyclic permutations of a restriction, to ensure that most recursive paths can continue for longer.
\section{Complexity Analysis}


We will now analyse the complexity of Algorithm \ref{alg:MatrixMults}, to confirm that it gives a different exponent than we would get from a straightforward Strassen-like algorithm.
Throughout this section, we will let $\omega_1, \omega_2, \omega_3$ be defined as the unique positive real numbers such that
\[
	n^{\omega_1-2}mp = \sum_{i=1}^q s_i n_i^{\omega_1-2}m_ip_i,
	\qquad nm^{\omega_2-2}p = \sum_{i=1}^q s_i n_im_i^{\omega_2-2}p_i,
	\qquad nmp^{\omega_3-2} = \sum_{i=1}^q s_i n_im_ip_i^{\omega_3-2}
\]
for a given tensor restriction $\Tens{n,m,p} \leq \bigoplus_{i=1}^q s_i \odot \Tens{n_i,m_i,p_i}$, and let $\omega_{\max}$ denote their maximum.
\begin{theorem}\label{thm:GeneralASI}
	Let $\rho = \min(\sum_{i=1}^q s_i \frac{n_im_i}{nm},\sum_{i=1}^q s_i
		\frac{m_ip_i}{mp},\sum_{i=1}^q s_i \frac{p_in_i}{pn})$,
	let $N_0\geq\lceil\left(\frac{\rho-1}{42\max(n,m,p)}\right)^\frac{1}{\omega_{\max}-3} \rceil$,
	and suppose
	that $\rho>1$,
	that $n_i < n$, $m_i < m$, $p_i < p$ for all~$i$,
	that $N_0> \frac{n-1}{\frac n{n_i}-1}$,
	$N_0> \frac{m-1}{\frac m{m_i}-1}$,
	$N_0> \frac{p-1}{\frac p{p_i}-1}$ for all~$i$,
	and that $\sum_{i=1}^q s_in_im_ip_i<nmp$.
	Then Algorithm \ref{alg:MatrixMults} takes
	\[
		\O(NMP(N^{\omega_1-3} + M^{\omega_2-3} + P^{\omega_3-3}))
	\]
	arithmetic operations.
\end{theorem}

Note that in the case each of the three factors in Romani's generalized ASI are identical and $n=m=p$ we exactly recover Romani's exponent.
Moreover, note that these requirements on the restriction are not very strong, and we conjecture that any non-trivial restriction with $\sum_{i=1}^q s_i n_im_ip_i < nmp$ will have these properties. In fact, we will later see a case in which we prove we can drop most of these conditions.
\begin{lemma}\label{lemma:1}
	Let $X\geq x,Y\geq y,Z\geq z \in \N$ and let $2 \leq \omega < 3 \in \R$.
	Then we have
	\begin{align*}
		\lceil \frac {X}{x} \rceil^{\omega-2} \lceil \frac{Y}{y} \rceil \lceil \frac{Z}{z} \rceil x^{\omega-2}yz \leq & X^{\omega-2}YZ + 7X^{\omega-3} \max(x,y,z)(XY+YZ+ZX).
	\end{align*}
\end{lemma}
\begin{proof}
	\begin{align*}
		 & \lceil \frac {X}{x} \rceil^{\omega-2} \lceil \frac{Y}{y} \rceil \lceil \frac{Z}{z} \rceil x^{\omega-2}yz \\
		 & \leq (X/x+1)^{\omega-2}(Y/y+1)(Z/z+1)x^{\omega-2}yz                                                      \\
		 & = X^{\omega-2}(1+x/X)^{\omega-2}(Y+y)(Z+z)                                                               \\
		 & \leq X^{\omega-2}(1+x/X)(Y+y)(Z+z)                                                                       \\
		 & \leq (X^{\omega-2}+xX^{\omega-3})(Y+y)(Z+z)                                                              \\
		 & \leq X^{\omega-2}YZ + X^{\omega-2}(yZ+Yz+yz) + xX^{\omega-3}(Y+y)(Z+z)                                   \\
		 & \leq X^{\omega -2}YZ + 7 X^{\omega-3} \max(x,y,z)(XY+YZ+ZX)
	\end{align*}
	\qed\end{proof}

\begin{proof}[Proof of Theorem \ref{thm:GeneralASI}]
	First, we show that $2 < \omega_i < 3$  for $i=1,2,3$. Consider the function $f(x) = \sum_{i=1}^q s_i \left( \frac{n_i}{n} \right) ^ {x-2} \frac{m_ip_i}{mp}$.
	Note that this function is non-increasing since its derivative is non-positive.
	Since $\rho > 1$ we know $f(2) > 1$, and since $\sum_{i=1}^q s_i n_i m_i p_i < nmp$ we know $f(3) < 1$.
	By definition, $f(\omega_1) = 1$ and hence $2 < \omega_1 < 3$. The same argument applies for $\omega_2$ and $\omega_3$.

	Let $T(N,M,P)$ denote the number of operations used by Algorithm \ref{alg:MatrixMults} to compute the matrix product of an $N\times M$ matrix with an $M \times P$ matrix.
	Note that if any of $N,M,P$ are at most $N_0$, we use the standard algorithm and so we have $T(N,M,P) \leq 2NMP$, and otherwise we have the recurrence
	\begin{align*}
		T(N,M,P) \leq {} & A \left( NM + MP + PN \right) +
		\sum_{i=1}^q s_i T(\lceil \frac{N}{n} \rceil n_i,\lceil \frac{M}{m} \rceil m_i, \lceil \frac{P}{p} \rceil p_i)
	\end{align*}

	Let $d=\frac{2A+84N_0^{\omega_{\max}-2}\max(n,m,p)}{\rho-1}$ and $L = 2N_0+d$.
	We show by induction on $V$ that for all positive integers $N,M,P$ with $NMP \leq V$ we have
	\[
		T(N,M,P) \leq LNMP(N^{\omega_1-3} + M^{\omega_2-3} + P^{\omega_3-3}) - d(NM+MP+PN).
	\]

	For the base case, we take $V = N_0^3$.
	We know that if $NMP \leq N_0^3$ then we have $\min(N,M,P) \leq N_0$ and so
	\begin{align*}
		T(N,M,P) & \leq 2NMP                                                             \\
		         & \leq 2N_0(NM+MP+PN)                                                   \\
		         & \leq (L-d)(NM+MP+PN)                                                  \\
		         & \leq LNMP(N^{\omega_1-3}+M^{\omega_2-3}+P^{\omega_3-3}) - d(NM+MP+PN)
	\end{align*}
	and so the statement holds for the base cases.
	Suppose we know for all positive integers $N,M,P$ with $NMP < V$ we have $T(N,M,P) \leq LNMP(N^{\omega_1-3} + M^{\omega_2-3} + P^{\omega_3-3}) - d(NM+MP+PN)$.
	Then we will show that for all positive integers $N,M,P$ with $NMP = V$ that we still have this inequality.
	First, consider that $\min(N,M,P) \leq N_0$. In this case, we have
	\begin{align*}
		T(N,M,P) & \leq 2NMP                                                              \\
		         & \leq 2N_0(NM+MP+PN)                                                    \\
		         & \leq (L-d)(NM+MP+PN)                                                   \\
		         & \leq LNMP(N^{\omega_1-3}+M^{\omega_2-3}+P^{\omega_3-3}) - d(NM+MP+PN).
	\end{align*}
	Otherwise, we know $\min(N,M,P) > N_0$.
	In this case, we can apply the recurrence formula, and the assumptions on $N_0$ ensure that we can apply the induction hypothesis to all recursive calls, which gives us
	\begin{alignat*}{3}
		T(N,M,P) & \leq A(NM+MP+PN) + \sum_{i=1}^q s_i T(\lceil \frac{N}{n} \rceil n_i,\lceil \frac{M}{m} \rceil m_i, \lceil \frac{P}{p} \rceil p_i)                                                                          \\
		         & \leq A(NM+MP+PN)                                                                                                                                                                                           \\
		         & \quad + L \lceil \frac{N}{n}\rceil ^{\omega_1-2}\lceil \frac{M}{m}\rceil \lceil \frac{P}{p}\rceil \sum_{i=1}^q s_i n_i^{\omega_1-2} m_i p_i                                                                \\
		         & \quad + L \lceil \frac{N}{n}\rceil \lceil \frac{M}{m}\rceil^{\omega_2-2} \lceil \frac{P}{p}\rceil \sum_{i=1}^q s_i n_i m_i^{\omega_2-2} p_i                                                                \\
		         & \quad + L \lceil \frac{N}{n}\rceil \lceil \frac{M}{m}\rceil \lceil \frac{P}{p}\rceil^{\omega_3-2} \sum_{i=1}^q s_i n_i m_i p_i^{\omega_3-2}                                                                \\
		         & \quad -d \lceil\frac{N}{n}\rceil \lceil \frac{M}{m}\rceil \sum_{i=1}^q s_i n_i m_i                                                                                                                         \\
		         & \quad -d \lceil\frac{M}{m}\rceil \lceil \frac{P}{p}\rceil \sum_{i=1}^q s_i m_i p_i                                                                                                                         \\
		         & \quad -d \lceil \frac{P}{p}\rceil \lceil \frac{N}{n}\rceil \sum_{i=1}^q s_i p_i n_i                                                         &       & \text{by induction hypothesis}                       \\
		         & \leq (A-d\rho)(NM+MP+PN)                                                                                                                                                                                   \\
		         & \quad + L \lceil \frac{N}{n}\rceil ^{\omega_1-2}\lceil \frac{M}{m}\rceil \lceil \frac{P}{p}\rceil n^{\omega_1-2}mp                                                                                         \\
		         & \quad + L \lceil \frac{N}{n}\rceil \lceil \frac{M}{m}\rceil^{\omega_2-2} \lceil \frac{P}{p}\rceil nm^{\omega_2-2}p                                                                                         \\
		         & \quad + L \lceil \frac{N}{n}\rceil \lceil \frac{M}{m}\rceil \lceil \frac{P}{p}\rceil^{\omega_3-2} nmp^{\omega_3-2}                          &       & \text{by definition of $\omega_1,\omega_2,\omega_3$} \\
		         & \leq \left(A-d\rho +7L\max(n,m,p)\left( N^{\omega_1-3} + M^{\omega_2-3} + P^{\omega_3-3} \right)\right)                                                                                                    \\
		         & \quad + L \left( N^{\omega_1-2}MP + NM^{\omega_2-2}P + NMP^{\omega_3-2}\right)                                                              &       & \text{by Lemma~\ref{lemma:1}}                        \\
		         & \leq \left( A-d\rho + 21L\max(n,m,p) N_0^{\omega_{\max}-3} \right)(NM+MP+PN)                                                                                                                               \\
		         & \quad + LNMP\left( N^{\omega_1-3} + M^{\omega_2-3} + P^{\omega_3-3} \right)                                                                 &       & \text{by $N,M,P\geq N_0$}                            \\
		         & \leq L NMP\left( N^{\omega_1-3} + M^{\omega_2-3} + P^{\omega_3-3} \right) - d(NM+MP+PN)                                                     & \quad & \text{by the choice of $d$}                          \\
	\end{alignat*}
	Hence, by induction, we are done.
	\qed\end{proof}

We provide a Mathematica Script that contains a step by step derivation and verification of the inequality chains in the proof of Theorem \ref{thm:GeneralASI} for specific values of the $\omega_i$, which can be found at
\begin{center}
	\url{https://github.com/mkauers/matrix-multiplication/structured}
\end{center}

In an earlier version of this manuscript we claimed that a variant of Algorithm~\ref{alg:MatrixMults} achieves an exponent $\omega_0$ with
\[
	(nmp)^{\frac{\omega_0}{3}} = \sum_{i=1}^q s_i (n_im_ip_i)^{\frac{\omega_0}{3}}.
\]
It has been pointed out to the authors by Alman and Li~\cite{LI2026} that the proof of this claim was flawed.
For Algorithm~\ref{alg:MatrixMults} as it is stated here we can show that the complexity stated in Theorem~\ref{thm:GeneralASI} is not only an upper bound but also a lower bound, so we have that Algorithm~\ref{alg:MatrixMults}
takes $\Theta(N^{\omega_1-2}MP + NM^{\omega_2-2}P + NMP^{\omega_3-2})$ operations to compute the product of an $N \times M$ matrix with an $M \times P$ matrix for any choice of $N_0$.
Of course, this does not imply that the statement
\[
	\Tens{n,m,p} \leq \bigoplus_{i=1}^q s_i \odot \Tens{n_i,m_i,p_i} \implies \omega \leq \omega_0
\]
is false, but only that Algorithm~\ref{alg:MatrixMults} does not achieve this exponent unless $\omega_0 = \omega_\mathrm{max}$.
\begin{theorem}
	Algorithm \ref{alg:MatrixMults} takes $\Omega(N^{\omega_1-2}MP + NM^{\omega_2-2}P + NMP^{\omega_3-2})$ operations to compute the product of an $N \times M$ matrix with an $M \times P$ matrix.
\end{theorem}

\begin{proof}
	Let $T(N,M,P)$ be defined as in the previous proof, so we know that for $\min(N,M,P) \leq N_0$ we use the standard algorithm and hence $T(N,M,P) \geq NMP$, and otherwise we have
	\[
		T(N,M,P) \geq \sum_{i=1}^q s_i T(\lceil \frac{N}{n}\rceil n_i,\lceil \frac{M}{m}\rceil m_i,\lceil \frac{P}{p}\rceil p_i).
	\]
	We show by induction on $V$ that
	\[
		T(N,M,P) \geq \frac{1}{3}(N^{\omega_1-2}MP+NM^{\omega_2-2}P+NMP^{\omega_3-2})
	\]
	for all positive integers $N,M,P$ with $NMP < V$.
	First, we look at the base case, $V=N_0^3$ and so $\min(N,M,P) \leq N_0$.
	Since we know $2 \leq \omega_i < 3$ we have $N^{\omega_1-3} + M^{\omega_2-3} + P^{\omega_3-3} \leq 3$ and hence we have
	\[
		T(N,M,P) \geq NMP \geq \frac{1}{3}NMP\left( N^{\omega_1-3} + M^{\omega_2-3} + P^{\omega_3-3} \right).
	\]
	Next, for the induction step, we suppose that the statement holds for $V-1$, and we aim to show that for all positive integers $N,M,P$ with $NMP=V$ the statement also holds.
	First, consider that $\min(N,M,P) \leq N_0$. In this case, we have $N^{\omega_1-3} + M^{\omega_2-3} + P^{\omega_3-3} \leq 3$ and hence we have
	\[
		T(N,M,P) \geq NMP \geq \frac{1}{3}NMP\left( N^{\omega_1-3} + M^{\omega_2-3} + P^{\omega_3-3} \right).
	\]
	Otherwise, $\min(N,M,P)>N_0$ then we can apply the recurrence relation, and therefore
	\begin{align*}
		T(N,M,P) & \geq \sum_{i=1}^q s_i T(\lceil \frac{N}{n}\rceil n_i,\lceil \frac{M}{m}\rceil m_i,\lceil \frac{P}{p}\rceil p_i)                                      \\
		         & \geq \frac{1}{3}\lceil\frac{N}{n}\rceil ^ {\omega_1-2} \lceil \frac{M}{m} \rceil \lceil \frac{P}{p}\rceil \sum_{i=1}^q s_i n_i^{\omega_1-2}m_ip_i    \\
		         & \quad + \frac{1}{3}\lceil\frac{N}{n}\rceil \lceil \frac{M}{m} \rceil^ {\omega_2-2} \lceil \frac{P}{p}\rceil \sum_{i=1}^q s_i n_im_i^ {\omega_2-2}p_i \\
		         & \quad + \frac{1}{3}\lceil\frac{N}{n}\rceil \lceil \frac{M}{m} \rceil \lceil \frac{P}{p}\rceil^ {\omega_3-2} \sum_{i=1}^q s_i n_im_ip_i^ {\omega_3-2} \\
		         & \geq \frac13 \left( N^{\omega_1-2}MP + NM^{\omega_2-2}P + NMP^{\omega_3-2}\right).
	\end{align*}
	Hence, by induction, we are done.
	\qed\end{proof}

We are particularly interested in the square case $N=M=P$ and from Algorithm \ref{alg:MatrixMults} we have that we can multiply $N\times N$ matrices in $\O(N^{\omega_{\max}})$, but of course this is likely not sensible.
A restriction of the form
\[\Tens{n,m,p} \leq \bigoplus_{i=1}^q s_i \odot \Tens{n_i,m_i,p_i}\]
also gives restrictions
\[
	\Tens{m,p,n} \leq \bigoplus_{i=1}^q s_i \odot \Tens{m_i,p_i,n_i}
	\quad\text{and}\quad
	\Tens{p,n,m} \leq \bigoplus_{i=1}^q s_i \odot \Tens{p_i,n_i,m_i}
\]
and so we can take the tensor product of these.
This gives us the restriction
\[
	\Tens{nmp,nmp,nmp} \leq \bigoplus_{i,j,k=1}^q s_is_js_k \odot \Tens{n_im_jp_k,n_km_ip_j,n_jm_kp_i},
\]
which has the following property.

\begin{definition}
	We say that a tensor restriction $\Tens{n,n,n} \leq \bigoplus_{i=1}^q s_i \odot \Tens{n_i,m_i,p_i}$ is cyclically invariant if the multiset of blocks $\{ \Tens{n_i, m_i, p_i} \}$ (accounting for multiplicities $s_i$) is invariant under the cyclic permutation of its dimensions $\Tens{a,b,c} \mapsto \Tens{b,c,a}$.
\end{definition}

\begin{proposition} \label{prop:symmetric_assumptions}
	Let $\Tens{n,n,n} \leq \bigoplus_{i=1}^q s_i \odot \Tens{n_i,m_i,p_i}$ be a non-trivial cyclically invariant tensor restriction. Suppose that $\omega_1 < 3$, and that $n_i \leq n$, $m_i \leq n$, and $p_i \leq n$ for all $i$. Then the conditions $\rho > 1$ and $\sum_{i=1}^q s_i n_i m_i p_i < n^3$ are satisfied.
\end{proposition}

In the above proposition, by ``non-trivial'' we mean that the restriction is not an isomorphism, as this will not give a terminating algorithm.

\begin{proof}
	Because the restriction is cyclically invariant, the equations defining
	$\omega_1, \omega_2,$ and $\omega_3$ are the
	exact same. Therefore $\omega_1 = \omega_2 = \omega_3 =: \tau$.

	Since we have such a restriction, there exist linear maps
	\[
		\phi_A\colon R^{n \times n} \rightarrow \bigoplus_{i=1}^q s_i \odot R^{n_i \times m_i},\quad
		\phi_B\colon R^{n \times n} \rightarrow \bigoplus_{i=1}^q s_i \odot R^{m_i \times p_i},\quad
		\phi_C\colon \bigoplus_{i=1}^q s_i \odot R^{p_i \times n_i} \rightarrow R^{n \times n}
	\]
	such that for the isomorphism $\Phi\colon A^\ast\otimes B^\ast\otimes C\to\operatorname{Bil}(A,B;C)$ we have
	\[
		\Phi(\<n,n,n>)(\A,\B) = \phi_C\left( \Phi\left(\bigoplus_{i=1}^q s_i \odot \<n_i,m_i,p_i>\right)(\phi_A(\A),\phi_B(\B)) \right)
	\]
	for all $\A \in R^{n \times n}$ and $\B\in R^{n \times n}$.
	Since $\phi_C$ is surjective, we can see that $\dim \bigoplus_{i=1}^q s_i \odot R^{p_i \times n_i} \geq \dim R^{n \times n}$
	and hence $\sum_{i=1}^q s_i p_i n_i \geq n^2$. By cyclic invariance, this tells us that $\rho \geq 1$.

	Suppose for contradiction that $\rho = 1$, i.e $\sum_{i=1}^q s_i p_i n_i = n^2$.
	Since $\phi_C$ is a surjective map between two spaces of the same dimension, $\phi_C$ is an isomorphism.
	Suppose, again for contradiction, that $\phi_A$ is not injective, i.e there exists a non-zero $\A \in R^{n \times n}$
	such that $\phi_A(\A) = 0$. Then we have that
	\[
		\A I = \Phi(\<n,n,n>)(\A,I) = \phi_C\left( \Phi\left(\bigoplus_{i=1}^q s_i \odot \<n_i,m_i,p_i>\right)(0, \phi_B(I)) \right) = \phi_C(0) = 0,
	\]
	which is a contradiction. Therefore, $\phi_A$ must be injective, and since the codomain and domain have the same dimension, $\phi_A$ must be an isomorphism.
	Similarly, $\phi_B$ must also be an isomorphism, showing that this restriction is trivial.
	Thus, we must have $\rho > 1$.

	Finally, we establish the bound $\sum_{i=1}^q s_i n_i m_i p_i <
		n^3$.
	We have
	\[
		n^\tau = \sum_{i=1}^q s_i n_i^{\tau-2} m_i p_i
	\]
	Multiplying by $n^{3-\tau}$ gives
	\[
		n^3 = \sum_{i=1}^q s_i n_i^{\tau-2} m_i p_i n^{3-\tau} \geq \sum_{i=1}^q s_i n_i m_i p_i
	\]
	where the inequality follows from the fact that $3-\tau > 0$ and $n_i \leq n$ for all $i$.  In the case of equality we would have $n_i = n$ for all $i$, and hence
	\[
		n^\tau = \sum_{i=1}^q s_i n_i^{\tau-2} m_i p_i = \sum_{i=1}^q s_i n^{\tau-2} m_i p_i
	\]
	which implies $\sum_{i=1}^q s_i m_i p_i = n^2$ and thus $\rho = 1$, contradicting our previous conclusion that $\rho > 1$.
	Therefore, we must have $\sum_{i=1}^q s_i n_i m_i p_i < n^3$.
	\qed\end{proof}

Proposition \ref{prop:symmetric_assumptions} shows that we can apply Theorem \ref{thm:GeneralASI} to this symmetrized restriction, so Algorithm~\ref{alg:MatrixMults} takes $\O(N^{\omega_{\mathrm{sym}}})$ operations to multiply $N\times N$ matrices using this restriction on $\Tens{nmp,nmp,nmp}$, where
\[
	(nmp)^{\omega_{\mathrm{sym}}} = \sum_{i,j,k=1}^q s_is_js_k (n_im_jp_k)^{\omega_{\mathrm{sym}}-2}n_km_ip_jn_jm_kp_i.
\]
Note that the exponent $\omega_\mathrm{sym}$ is exactly the $\omega_0$ from Romani's generalized ASI.

We can also give an estimate of the leading coefficients for Algorithm \ref{alg:MatrixMults} by assuming all recursive steps can be done exactly down to a base case, which is with $N_0=1$, as is done to get the leading coefficient of 7 with Strassen's algorithm.
In this case, we have that the number of operations is bounded by
$L_1 N^{\omega_1-2}MP+L_2 NM^{\omega_2-2}P + L_3 NMP^{\omega_3-2}$ where
\[
	L_1 = \frac{1}{3}+ \frac{A_2}{-mp+\sum_i s_i m_ip_i},\quad
	L_2 = \frac{1}{3}+ \frac{A_3}{-pn+\sum_i s_ip_in_i},\quad
	L_3 = \frac{1}{3}+ \frac{A_1}{-nm+\sum_i s_in_im_i},
\]
and here $A_1$ is the number of arithmetic operations of the $n\times m$ blocks, $A_2$ the number of arithmetic operations of the $m \times p$ blocks and $A_3$ the number of arithmetic operations of the $p \times n$ blocks.
In particular, this gives us an algorithm to multiply $N \times N$ matrices with exponent $\omega_{\max}$ and leading coefficient $\leq L_1 + L_2 + L_3$.
In the special case $\Tens{2,2,2} \leq 7 \odot \Tens{1,1,1}$ using Strassen's algorithm with $A_1+A_2+A_3 = 18$ this gives the leading coefficient 7 as expected.

Now we can apply this to the cyclically invariant restriction on $\<nmp,nmp,nmp>$ we constructed. Suppose the restriction on $\<n,m,p>$ requires $A_1,A_2,A_3$ arithmetic operations, the restriction on $\<m,p,n>$ requires $B_1,B_2,B_3$ arithmetic operations, and the restriction on $\<p,n,m>$ requires $C_1,C_2,C_3$ arithmetic operations, and let $R = \sum_{i=1}^q s_i n_i m_i p_i$. Then the cyclically invariant restriction on $\<nmp,nmp,nmp>$ we constructed uses Algorithm \ref{alg:MatrixMults}  to multiply $N\times N$ with exponent $\omega_\mathrm{sym}$ and leading coefficient at most
\begin{equation}\label{eq:leadingcoeff}
	1 + \frac{nmp(A_1p + A_2n + A_3m) + R(B_1pn + B_2nm + B_3mp) + R^2(C_1+C_2+C_3)}{-(nmp)^2 + \sum_{i,j,k=1}^q s_in_im_i\,s_jm_jp_j\,s_kn_kp_k}.
\end{equation}
This leading coefficient is slightly crudely bounded and one can achieve a
slightly better coefficient, though this isn't as easy to write down. Note also
that our bounds on the leading coefficients only hold under the assumption that all
divisions can be done exactly. This is, of course, not true in general. For
example, a $17\times 17$ matrix multiplication problem takes more than $7\cdot
	17^{\log_27}$ arithmetic operations using Strassen's algorithm, due to the need
to zero pad. We compute our algorithms' leading coefficients with this
assumption as to better compare them to existing bounds.

\section{Specific Decompositions}

We have already used the Moosbauer and Poole's tensor decomposition of
$\Tens{6,6,6}$ in the form $137\odot\Tens{1}\oplus8\odot\Tens{2,1,1}$ as
illustrating example. According to Theorem~\ref{thm:GeneralASI}, this
decomposition leads to an exponent $2.805065$, slightly better than the exponent
$2.80754$ announced by Moosbauer and Poole and than Strassen's exponent
$2.80735$, though not as good as the exponent $2.7925$ of the decompositions of
Novikov et al.~\cite{ALPHAEVOLVE2025} and Dumas, Pernet, and
Sedoglavic~\cite{DUMAS2025}. We have repeated the computation of Moosbauer and
Pool in order to generate further decompositions of $\Tens{6,6,6}$ in the hope
to find some that contain more copies of $\Tens{2,1,1}$. The best we found
contains 18 such copies and thus leads to an exponent~$2.8019$. The leading
coefficient for this decomposition turns out to be~$7.67$, compared to $7$ for
Strassen's algorithm.

We have also repeated the computation of Moosbauer and Poole for $\Tens{5,5,5}$
in the hope to find a decomposition with a better structure.
The best we found leads to an exponent $2.8091$, slightly worse than Strassen's exponent.

In a next step, for various choices of $n,m,p$, we have searched for
decompositions of $\Tens{n,m,p}$ that contain copies of $\Tens{1,2,2}$ and/or $\Tens{k,1,1}$
for some~$k$, or permutations of these.
This search procedure is described in more detail below.
The results are summarized in Table~\ref{tab:1}.
In this table, $\omega_{\text{rank}}$ refers to the exponent reported in Sedoglavic's table~\cite{SEDOGLAVICTABLE},
$\omega_{\mathrm{sym}}$ refers to the exponent obtained via Theorem~\ref{thm:GeneralASI},
$L$~is the leading coefficient given by~\eqref{eq:leadingcoeff} where the various addition counts were determined
by the software of M\aa rtensson and Wagner~\cite{MARTENSSON2025}.
\begin{table}
	\begin{center}
		\begin{tabular}{@{}cc|ccl@{}}
			$nmp$ & $\omega_{\mathrm{rank}}$ & $\omega_{\mathrm{sym}}$ & $L$  & structure                                                                                                             \\\hline
			336   & 2.77430                  & 2.77430                 & n/a  & $40\odot\<1,1,1>$                                                                                                     \\
			444   & 2.79248                  & 2.79248                 & n/a  & $48\odot\<1,1,1>$                                                                                                     \\
			346   & 2.79820                  & 2.79820                 & 6.94 & $54\odot\<1,1,1>$                                                                                                     \\\hline
			666   & 2.80754                  & 2.80190                 & 7.67 & $6\odot\<1,1,2>\oplus6\odot\<2,1,1>\oplus6\odot\<1,2,1>\oplus117\odot\<1,1,1>$                                        \\
			337   & 2.81803                  & 2.80525                 & 8.88 & $10\odot\<1,1,2>\oplus29\odot\<1,1,1>$                                                                                \\
			567   & 2.81122                  & 2.80547                 & 7.57 & $6\odot\<1,1,2>\oplus2\odot\<1,1,3>\oplus2\odot\<3,1,1>\oplus3\odot\<1,2,1>\oplus120\odot\<1,1,1>$                    \\
			566   & 2.81200                  & 2.80566                 & 7.67 & $5\odot\<1,1,2>\oplus\<1,1,3>\oplus\<3,1,1>\oplus5\odot\<1,2,1>\oplus\<1,3,1>\oplus101\odot\<1,1,1>$                  \\
			568   & 2.81124                  & 2.80626                 & 7.60 & $13\odot\<1,1,2>\oplus\<1,1,3>\oplus3\odot\<1,2,1>\oplus135\odot\<1,1,1>$                                             \\
			556   & 2.81430                  & 2.80643                 & 7.81 & $4\odot\<1,1,2>\oplus2\odot\<2,1,1>\oplus2\odot\<3,1,1>\oplus2\odot\<1,2,1>\oplus2\odot\<1,3,1>\oplus82\odot\<1,1,1>$ \\\hline
			222   & 2.80735                  & 2.80735                 & 7.00 & $7\odot\<1,1,1>$                                                                                                      \\\hline
			557   & 2.81378                  & 2.80831                 & 7.57 & $8\odot\<1,1,2>\oplus\<1,1,3>\oplus\<2,1,1>\oplus3\odot\<1,2,1>\oplus100\odot\<1,1,1>$                                \\
			558   & 2.81400                  & 2.80855                 & 7.76 & $8\odot\<1,1,2>\oplus2\odot\<1,1,3>\oplus2\odot\<2,1,1>\oplus2\odot\<1,2,1>\oplus114\odot\<1,1,1>$                    \\
			555   & 2.81626                  & 2.80911                 & 7.55 & $3\odot\<1,1,2>\oplus\<1,1,3>\oplus\<3,1,1>\oplus3\odot\<1,2,1>\oplus\<1,3,1>\oplus72\odot\<1,1,1>$                   \\
			457   & 2.81954                  & 2.81152                 & 7.78 & $10\odot\<1,1,2>\oplus2\odot\<1,1,3>\oplus2\odot\<1,2,1>\oplus74\odot\<1,1,1>$                                        \\
			467   & 2.81746                  & 2.81199                 & 7.58 & $10\odot\<1,1,2>\oplus2\odot\<2,1,1>\oplus2\odot\<1,2,1>\oplus95\odot\<1,1,1>$                                        \\
			234   & 2.82789                  & 2.81214                 & 7.69 & $4\odot\<1,1,2>\oplus12\odot\<1,1,1>$                                                                                 \\
			458   & 2.82001                  & 2.81275                 & 7.99 & $12\odot\<1,1,2>\oplus2\odot\<1,1,3>\oplus\<1,2,1>\oplus86\odot\<1,1,1>$                                              \\
			237   & 2.85366                  & 2.81336                 & 12.4 & $11\odot\<1,1,2>\oplus4\odot\<1,1,3>\oplus\<1,1,1>$                                                                   \\
			334   & 2.81899                  & 2.81359                 & 6.47 & $\<1,1,3>\oplus26\odot\<1,1,1>$                                                                                       \\
			577   & 2.81962                  & 2.81366                 & 7.86 & $9\odot\<1,1,2>\oplus\<1,1,5>\oplus\<3,1,1>\oplus7\odot\<1,2,1>\oplus136\odot\<1,1,1>$                                \\
		\end{tabular}
	\end{center}
	\caption{New decompositions found for various matrix multiplication tensors $\<n,m,p>$.
		The rows for 336, 444, 367, and 222 are included for comparison only.
		For 222 and 367, we report the leading coefficients resulting from our computations,
		although for 222 it is known that the leading coefficient can be reduced to~$6$,
		which suggests that for the other leading coefficients there may also still be room
		for improvement.
		For 336 and 444, the tool of M\aa rtensson and Wagner~\cite{MARTENSSON2025} is not applicable because these decompositions
		involve rational coefficients. It is known though~\cite{DUMAS2025} that for 444, a leading
		coefficient of $7$ can be achieved.
	}\label{tab:1}
\end{table}

It turned out that besides 666, there are several other formats where we obtain an exponent smaller
than Strassen's, but none of them beats the currently best known exponents for 336, 444, or 346.
The decompositions claimed in the table are available electronically at
\begin{center}
	\url{https://github.com/mkauers/matrix-multiplication/structured}
\end{center}
The following corollary applies to formats like 228, which can be factorized as
tensor product (e.g., $\<2,2,8>=\<2,2,2>\otimes\<1,1,4>$). Our search procedure has
encountered several of these, but since they are predictable, they are not included
in the table.

\begin{corollary}
	Given restrictions $\phi,\phi'$ with $\<n,m,p> \leq_{\phi} \bigoplus_{i=1}^q s_i \odot \<n_i,m_i,p_i>$
	and $\<n',m',p'> \leq_{\phi'} \bigoplus_{i=1}^{q'} s_i' \odot \<n_i',m_i',p_i'>$ which, when symmetrized,
	achieve exponents $\omega,\omega'$ respectively
	in Algorithm \ref{alg:MatrixMults}. Then there exists a restriction $\phi''$ on $\<nn',mm',pp'>$
	which, when symmetrized, achieves exponent $\omega'' = \min(\omega,\omega')$ in Algorithm \ref{alg:MatrixMults}.
\end{corollary}
\begin{proof}
	Without loss of generality, suppose $\omega\leq\omega'$. Then we can take $\phi''$ to give the restriction
	\begin{align*}
		\<nn',mm',pp'> \cong \<n,m,p> \otimes \<n',m',p'> & \leq_{\phi''} \left( \bigoplus_{i=1}^q s_i \odot \<n_i,m_i,p_i> \right) \otimes \<n',m',p'> \\
		                                                  & \cong \bigoplus_{i=1}^q s_i \odot \<n_in',m_im',p_ip'>,
	\end{align*}
	given via the tensor product of $\phi$ with the identity. Applying Theorem \ref{thm:GeneralASI} after symmetrizing, we get exponent $\omega''$ defined by
	\[
		(nmp)^{\omega''}(n'm'p')^{\omega''} = \sum_{i,j,k=1}^q s_is_js_k (n_im_jp_k)^{\omega''-2}n_km_ip_jn_jm_kp_i(n'm'p')^{\omega''}.
	\]
	Note that by dividing both sides by $(n'm'p')^{\omega''}$ we obtain the defining equation for~$\omega$.
	Therefore, $\omega'' = \omega$.
	\qed
\end{proof}

In order to find these decompositions, we proceeded according to the following steps.

\textbf{Step 1.} 
In order to obtain a decomposition containing $\<1,2,2>$, we first apply a flip graph
search~\cite{KAUERS2022,KAUERS2023,MOOSBAUER2025,KAUERS2025,PERMINOV2025} to the tensor
\begin{alignat*}1
	 & \sum_{i=1}^n\sum_{j=1}^m\sum_{k=1}^p a_{i,j}\otimes b_{j,k}\otimes c_{k,i}
	- a_{1,1}\otimes b_{1,1} \otimes c_{1,1}
	- a_{2,1}\otimes b_{1,1} \otimes c_{1,2}
	- a_{1,2}\otimes b_{2,1} \otimes c_{1,1}
	- a_{2,2}\otimes b_{2,1} \otimes c_{1,2}
\end{alignat*}
in order to find a decomposition with as low rank as possible.
Adding
\begin{alignat*}1
	a_{1,1}\otimes b_{1,1} \otimes c_{1,1}
	+ a_{2,1}\otimes b_{1,1} \otimes c_{1,2}
	+ a_{1,2}\otimes b_{2,1} \otimes c_{1,1}
	+ a_{2,2}\otimes b_{2,1} \otimes c_{1,2}
\end{alignat*}
to the result yields a decomposition of $\<n,m,p>$ which contains a copy of $\<1,2,2>$.
Compared to the computation time invested into this first step, the computation time of
all the subsequent steps is negligible.
Alas, while we did find decompositions containing $\<1,2,2>$ for certain formats matching
the best known rank for the format, none of them appears in Table~\ref{tab:1} because they
all were outperformed by decompositions involving copies of $\<1,1,k>$.

\textbf{Step 2.} 
Flip graph searches have so far only been employed in order to find decompositions of low rank.
The same technique can however also be used to optimize other features of a decomposition.
On the decompositions obtained from step~1, as well as on the decompositions computed in~\cite{KAUERS2025},
we performed a flip graph search with the aim of
maximizing the number of copies of $\<1,1,k>$, $\<1,\ell,1>$, $\<m,1,1>$ contained in it.
Note that these patterns are quite easy to detect. They amount to two components $a\otimes b\otimes c$
which have one factor in common.
It must be noted however that in order to apply Algorithm~\ref{alg:MatrixMults}, we must use a selection
of copies of $\<1,1,k>$, $\<1,\ell,1>$, $\<m,1,1>$ that do not overlap.
For example,
\[
	a\otimes b\otimes c
	+ a\otimes b'\otimes c'
	+ a''\otimes b'\otimes c''
\]
contains a copy of $\<2,1,1>$ (because $a$ appears twice) as well as a copy of $\<1,2,1>$
(because $b'$ appears twice), but we must not use both of them because they overlap in
$a\otimes b'\otimes c'$.

\textbf{Step 3.} 
To the decompositions obtained in Step~2, we next apply a number of random elements of
de Groote's symmetry group~\cite{DEGROOTE19781,DEGROOTE19782} in search for an orbit element with a small support.
In principle, it would also be possible to minimize support with a flip graph search,
but this approach would likely destroy the structure constructed in steps 1 and~2.
In contrast, the structure is invariant under symmetries, and therefore preserved by
the application of elements of the symmetry group.

\textbf{Step 4.} 
Up to this point, all computations are done for the field $K=\set Z_2$.
In the next step,
we lift the coefficients of the decomposition to integers. There are various ways for
doing this.
Hensel lifting~\cite{KAUERS2023,VONZURGATHEN1999} was used in earlier papers about flip graph
searches.
However, as pointed out by Kemper~\cite{KEMPER2025}, the resulting coefficients tend
to be more complicated than necessary.
In several instances where Hensel lifting led
to decompositions involving rational numbers with rather lengthy numerators and denominators,
he was able to obtain a decomposition involving only the coefficients $-1,0,1$, starting
from the same decomposition over~$\set Z_2$.
More critically, we must take care that the structure imposed on the decomposition during
steps 1 and~2 is preserved during lifting.
This is not automatically ensured. For example,
the decomposition of $\<6,6,6>$ presented by Moosbauer and Poole contains two copies of
$\<1,1,1>$ which become a copy of $\<1,1,2>$ when coefficients are taken modulo~$2$.
We continue to use Hensel lifting and address this issue by imposing additional constraints
in order to ensure that the structure of the decomposition is preserved.
In the present
context, Hensel lifting leads to an underdetermined linear system over~$\set Z_2$, and the
additional constraints can be easily encoded as additional linear equations which we append
to this system.
In the same way, we try to preserve the sparsity from Step~3 by imposing
additional constraints so as to ensure that every zero in $\set Z_2$ will remain a zero
during the lifting.
This may seem somewhat aggressive, but it worked surprisingly well and
in many cases led to decompositions involving only the coefficients $-1,0,1$.
In Table~\ref{tab:2}, we list some of the that we have not been able to lift from $\set Z_2$ to~$\set Q$.

\textbf{Step 5.} 
Finally, we determine for each decomposition the number of additions needed to execute it.
For this step, we employ software provided by M\aa rtensson and Wagner ~\cite{MARTENSSON2025}.

\begin{table}
	\begin{center}
		\begin{tabular}{@{}cc|cl@{}}
			$nmp$ & $\omega_{\mathrm{rank}}$ & $\omega_{\mathrm{sym}}$ & structure                                                                             \\\hline
			444   & 2.77729                  & 2.77729                 & $47\odot\<1,1,1>$                                                                     \\\hline
			455   & 2.79498                  & 2.79153                 & $2\odot\<2,1,1>\oplus\<3,1,1>\oplus66\odot\<1,1,1>$                                   \\
			445   & 2.80305                  & 2.79290                 & $3\odot\<1,1,2>\oplus3\odot\<1,3,1>\oplus45\odot\<1,1,1>$                             \\
			447   & 2.82462                  & 2.79833                 & $38\odot\<1,1,2>\oplus9\odot\<1,1,1>$                                                 \\
			338   & 2.81107                  & 2.79855                 & $11\odot\<1,1,2>\oplus33\odot\<1,1,1>$                                                \\
			446   & 2.81998                  & 2.80121                 & $21\odot\<1,1,2>\oplus\<1,1,3>\oplus28\odot\<1,1,1>$                                  \\
			456   & 2.81273                  & 2.80401                 & $\<1,1,2>\oplus\<3,1,1>\oplus2\odot\<1,2,1>\oplus4\odot\<1,3,1>\oplus68\odot\<1,1,1>$ \\
			356   & 2.81312                  & 2.80646                 & $8\odot\<1,1,2>\oplus52\odot\<1,1,1>$                                                 \\
			344   & 2.81896                  & 2.80674                 & $7\odot\<1,2,1>\oplus24\odot\<1,1,1>$                                                 \\
			333   & 2.85405                  & 2.83686                 & $2\odot\<1,2,1>\oplus15\odot\<1,1,1>\oplus\<1,2,2>$
		\end{tabular}
	\end{center}
	\caption{New decompositions found for various matrix multiplication tensors $\<n,m,p>$, only valid for $K=\set Z_2$.
		The row for 444 is included for comparison only.
		The structure for 333 is noteworthy because it contains a copy of $\<1,2,2>$.
		Without prescribing this component, the best decomposition we find for 333 is $4\odot\<1,2,1>\oplus15\otimes\<1,1,1>$,
		which gives rise to the slightly larger exponent $2.84297$.
	}\label{tab:2}
\end{table}

\section{Simulation}

Since the bounds on the leading coefficient do not fully reflect the actual operation count of the algorithms, we simulate the recursive calls performed by our algorithm for different input sizes $N$ in order to get a more accurate idea of the runtime of the algorithm.
Of course this approach does not take into account memory access costs and other practical considerations, but it does give a better estimate of the operation count than just the leading coefficient.
We compare Winograd's variant of Strassen's algorithm, which has a leading coefficient of~$6$, to our algorithm using the decomposition of $\Tens{6,6,6}$ we found.
The blue $\times$ and orange $+$ in Figure~\ref{fig:sim}
show the simulated operation counts for the algorithms, where for Moosbauer-Poole, we switch to Winograd-Strassen for $N<10^4$ and in both cases switch to the standard algorithm for $N<35$ as one would do in an actual implementation.
The lines show the complexity estimate using the leading coefficient according to the formulas given above.
The reason that our simulation shows lower operation counts than the reported bounds is that we switch to more efficient algorithms for small matrix sizes, while in the analysis we assumed that the recursion is performed exactly all the way down to $1\times 1$ matrices.
We can see that the results from our simulation do not form a smooth curve.
This is due to the necessary zero padding, which makes the algorithms sensitive to the input size.
In the simulations we start to see improvements over Strassen's algorithm around matrix sizes of about $10^6$ and a consistent outperformance starting at $10^{10}$.
\begin{figure}
	\begin{center}
		\includegraphics[scale = 0.52]{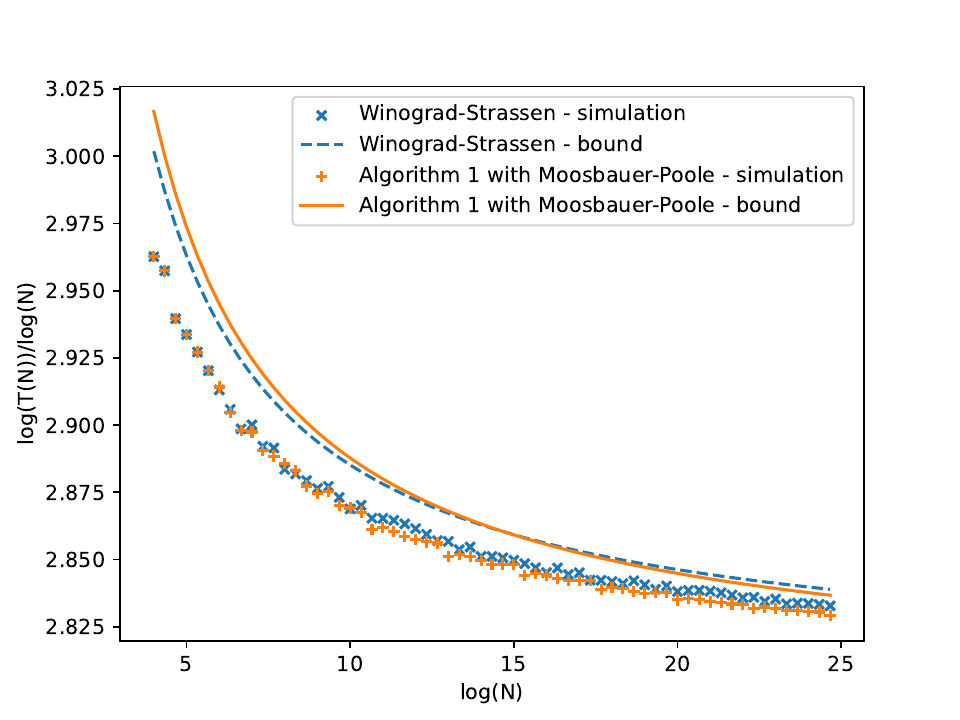}
	\end{center}
	\caption{Simulated operation count for our algorithm and Strassen's algorithm}
	\label{fig:sim}
\end{figure}

\bibliographystyle{elsarticle-num}
\bibliography{main}

\begin{thebibliography}{10}
\expandafter\ifx\csname url\endcsname\relax
  \def\url#1{\texttt{#1}}\fi
\expandafter\ifx\csname urlprefix\endcsname\relax\def\urlprefix{URL }\fi
\expandafter\ifx\csname href\endcsname\relax
  \def\href#1#2{#2} \def\path#1{#1}\fi

\bibitem{STRASSEN1969}
V.~STRASSEN, \href{http://eudml.org/doc/131927}{Gaussian elimination is not
  optimal.}, Numerische Mathematik 13 (1969) 354--356.
\newline\urlprefix\url{http://eudml.org/doc/131927}

\bibitem{BINI1979}
D.~Bini, M.~Capovani, F.~Romani, G.~Lotti, O(n2.7799) complexity for n × n
  approximate matrix multiplication, Information Processing Letters 8 (06
  1979).
\newblock \href {https://doi.org/10.1016/0020-0190(79)90113-3}
  {\path{doi:10.1016/0020-0190(79)90113-3}}.

\bibitem{SCHONHAGE1981}
A.~Sch\"{o}nhage, \href{https://doi.org/10.1137/0210032}{Partial and total
  matrix multiplication}, SIAM J. Comput. 10~(3) (1981) 434–455.
\newblock \href {https://doi.org/10.1137/0210032} {\path{doi:10.1137/0210032}}.
\newline\urlprefix\url{https://doi.org/10.1137/0210032}

\bibitem{STRASSEN1986}
V.~Strassen, The asymptotic spectrum of tensors and the exponent of matrix
  multiplication, in: 27th Annual Symposium on Foundations of Computer Science
  (sfcs 1986), 1986, pp. 49--54.
\newblock \href {https://doi.org/10.1109/SFCS.1986.52}
  {\path{doi:10.1109/SFCS.1986.52}}.

\bibitem{COPPERSMITH1990}
D.~Coppersmith, S.~Winograd,
  \href{https://www.sciencedirect.com/science/article/pii/S0747717108800132}{Matrix
  multiplication via arithmetic progressions}, Journal of Symbolic Computation
  9~(3) (1990) 251--280, computational algebraic complexity editorial.
\newblock \href {https://doi.org/https://doi.org/10.1016/S0747-7171(08)80013-2}
  {\path{doi:https://doi.org/10.1016/S0747-7171(08)80013-2}}.
\newline\urlprefix\url{https://www.sciencedirect.com/science/article/pii/S0747717108800132}

\bibitem{ALMAN2024}
J.~Alman, R.~Duan, V.~V. Williams, Y.~Xu, Z.~Xu, R.~Zhou,
  \href{https://arxiv.org/abs/2404.16349}{More asymmetry yields faster matrix
  multiplication} (2024).
\newblock \href {http://arxiv.org/abs/2404.16349} {\path{arXiv:2404.16349}}.
\newline\urlprefix\url{https://arxiv.org/abs/2404.16349}

\bibitem{KAUERS2022}
M.~Kauers, J.~Moosbauer, \href{https://arxiv.org/abs/2210.04045}{The
  fbhhrbnrssshk-algorithm for multiplication in $\mathbb{Z}_2^{5\times5}$ is
  still not the end of the story} (2022).
\newblock \href {http://arxiv.org/abs/2210.04045} {\path{arXiv:2210.04045}}.
\newline\urlprefix\url{https://arxiv.org/abs/2210.04045}

\bibitem{KAUERS2023}
M.~Kauers, J.~Moosbauer, \href{https://doi.org/10.1145/3597066.3597120}{Flip
  graphs for matrix multiplication}, in: Proceedings of the 2023 International
  Symposium on Symbolic and Algebraic Computation, ISSAC '23, Association for
  Computing Machinery, New York, NY, USA, 2023, p. 381–388.
\newblock \href {https://doi.org/10.1145/3597066.3597120}
  {\path{doi:10.1145/3597066.3597120}}.
\newline\urlprefix\url{https://doi.org/10.1145/3597066.3597120}

\bibitem{KAUERS2025}
M.~Kauers, I.~Wood, \href{https://arxiv.org/abs/2510.19787}{Exploring the meta
  flip graph for matrix multiplication} (2025).
\newblock \href {http://arxiv.org/abs/2510.19787} {\path{arXiv:2510.19787}}.
\newline\urlprefix\url{https://arxiv.org/abs/2510.19787}

\bibitem{MOOSBAUER2023}
M.~Kauers, J.~Moosbauer, \href{https://arxiv.org/abs/2306.00882}{Some new
  non-commutative matrix multiplication algorithms of size $(n,m,6)$} (2023).
\newblock \href {http://arxiv.org/abs/2306.00882} {\path{arXiv:2306.00882}}.
\newline\urlprefix\url{https://arxiv.org/abs/2306.00882}

\bibitem{MOOSBAUER2025}
J.~Moosbauer, M.~Poole, \href{https://doi.org/10.1145/3747199.3747566}{Flip
  graphs with symmetry and new matrix multiplication schemes}, in: Proceedings
  of the 2025 International Symposium on Symbolic and Algebraic Computation,
  ISSAC '25, Association for Computing Machinery, New York, NY, USA, 2025, p.
  233–239.
\newblock \href {https://doi.org/10.1145/3747199.3747566}
  {\path{doi:10.1145/3747199.3747566}}.
\newline\urlprefix\url{https://doi.org/10.1145/3747199.3747566}

\bibitem{PERMINOV2025}
A.~I. Perminov, \href{https://arxiv.org/abs/2511.20317}{Fast matrix
  multiplication via ternary meta flip graphs} (2025).
\newblock \href {http://arxiv.org/abs/2511.20317} {\path{arXiv:2511.20317}}.
\newline\urlprefix\url{https://arxiv.org/abs/2511.20317}

\bibitem{WOOD2025}
M.~Kauers, I.~Wood, \href{https://arxiv.org/abs/2505.05896}{Consequences of the
  moosbauer-poole algorithms} (2025).
\newblock \href {http://arxiv.org/abs/2505.05896} {\path{arXiv:2505.05896}}.
\newline\urlprefix\url{https://arxiv.org/abs/2505.05896}

\bibitem{SMIRNOV2013}
A.~Smirnov, The bilinear complexity and practical algorithms for matrix
  multiplication, Computational Mathematics and Mathematical Physics 53 (12
  2013).
\newblock \href {https://doi.org/10.1134/S0965542513120129}
  {\path{doi:10.1134/S0965542513120129}}.

\bibitem{KAPORIN2024}
I.~Kaporin,
  \href{https://journals.eco-vector.com/2686-9543/article/view/647987}{Semi-analytical
  solution of brent equations}, Doklady Mathematics 518~(1) (2024) 29--34.
\newline\urlprefix\url{https://journals.eco-vector.com/2686-9543/article/view/647987}

\bibitem{NOVIKOV2025}
A.~Novikov, N.~Vũ, M.~Eisenberger, E.~Dupont, P.-S. Huang, A.~Z. Wagner,
  S.~Shirobokov, B.~Kozlovskii, F.~J.~R. Ruiz, A.~Mehrabian, M.~P. Kumar,
  A.~See, S.~Chaudhuri, G.~Holland, A.~Davies, S.~Nowozin, P.~Kohli, M.~Balog,
  \href{https://arxiv.org/abs/2506.13131}{Alphaevolve: A coding agent for
  scientific and algorithmic discovery} (2025).
\newblock \href {http://arxiv.org/abs/2506.13131} {\path{arXiv:2506.13131}}.
\newline\urlprefix\url{https://arxiv.org/abs/2506.13131}

\bibitem{FAWZI2022}
A.~Fawzi, M.~Balog, A.~Huang, T.~Hubert, B.~Romera-Paredes, M.~Barekatain,
  A.~Novikov, F.~J.~R. Ruiz, J.~Schrittwieser, G.~Swirszcz, D.~Silver,
  D.~Hassabis, P.~Kohli, Discovering faster matrix multiplication algorithms
  with reinforcement learning, Nature 610~(7930) (2022) 47--53.
\newblock \href {https://doi.org/10.1038/s41586-022-05172-4}
  {\path{doi:10.1038/s41586-022-05172-4}}.

\bibitem{PROBERT1976}
R.~L. Probert, \href{https://doi.org/10.1137/0205016}{On the additive
  complexity of matrix multiplication}, SIAM Journal on Computing 5~(2) (1976)
  187--203.
\newblock \href {https://doi.org/10.1137/0205016} {\path{doi:10.1137/0205016}}.
\newline\urlprefix\url{https://doi.org/10.1137/0205016}

\bibitem{KARSTADT2017}
E.~Karstadt, O.~Schwartz, \href{https://doi.org/10.1145/3087556.3087579}{Matrix
  multiplication, a little faster}, in: Proceedings of the 29th ACM Symposium
  on Parallelism in Algorithms and Architectures, SPAA '17, Association for
  Computing Machinery, New York, NY, USA, 2017, p. 101–110.
\newblock \href {https://doi.org/10.1145/3087556.3087579}
  {\path{doi:10.1145/3087556.3087579}}.
\newline\urlprefix\url{https://doi.org/10.1145/3087556.3087579}

\bibitem{BENIAMINI2019}
G.~Beniamini, O.~Schwartz,
  \href{https://doi.org/10.1145/3323165.3323188}{Faster matrix multiplication
  via sparse decomposition}, in: The 31st ACM Symposium on Parallelism in
  Algorithms and Architectures, SPAA '19, Association for Computing Machinery,
  New York, NY, USA, 2019, p. 11–22.
\newblock \href {https://doi.org/10.1145/3323165.3323188}
  {\path{doi:10.1145/3323165.3323188}}.
\newline\urlprefix\url{https://doi.org/10.1145/3323165.3323188}

\bibitem{SCHWARTZ2024}
O.~Schwartz, S.~Toledo, N.~Vaknin, G.~Wiernik, Alternative basis matrix
  multiplication is fast and stable, in: 2024 IEEE International Parallel and
  Distributed Processing Symposium (IPDPS), 2024, pp. 38--51.
\newblock \href {https://doi.org/10.1109/IPDPS57955.2024.00013}
  {\path{doi:10.1109/IPDPS57955.2024.00013}}.

\bibitem{MARTENSSON2025}
E.~Mårtensson, P.~S. Wagner,
  \href{https://epubs.siam.org/doi/abs/10.1137/1.9781611978759.4}{The Number of
  the Beast: Reducing Additions in Fast Matrix Multiplication Algorithms for
  Dimensions up to 666}, 2025, pp. 47--60.
\newblock \href {https://doi.org/10.1137/1.9781611978759.4}
  {\path{doi:10.1137/1.9781611978759.4}}.
\newline\urlprefix\url{https://epubs.siam.org/doi/abs/10.1137/1.9781611978759.4}

\bibitem{PAN1984}
V.~Pan, How to multiply matrices faster, Springer-Verlag, Berlin, Heidelberg,
  1984.

\bibitem{SCHWARTZ2025}
O.~Schwartz, E.~Zwecher, \href{https://arxiv.org/abs/2508.01748}{Towards faster
  feasible matrix multiplication by trilinear aggregation} (2025).
\newblock \href {http://arxiv.org/abs/2508.01748} {\path{arXiv:2508.01748}}.
\newline\urlprefix\url{https://arxiv.org/abs/2508.01748}

\bibitem{ROMANI1981}
F.~Romani, \href{https://doi.org/10.1137/0211020}{Some properties of disjoint
  sums of tensors related to matrix multiplication}, SIAM Journal on Computing
  11~(2) (1982) 263--267.
\newblock \href {http://arxiv.org/abs/https://doi.org/10.1137/0211020}
  {\path{arXiv:https://doi.org/10.1137/0211020}}, \href
  {https://doi.org/10.1137/0211020} {\path{doi:10.1137/0211020}}.
\newline\urlprefix\url{https://doi.org/10.1137/0211020}

\bibitem{LADERMAN1976}
J.~D. Laderman, A noncommutative algorithm for multiplying $3 \times 3$
  matrices using 23 multiplications, Bulletin of the American Mathematical
  Society 82~(1) (1976) 126--128.
\newblock \href {https://doi.org/10.1090/S0002-9904-1976-13988-2}
  {\path{doi:10.1090/S0002-9904-1976-13988-2}}.

\bibitem{LI2026}
J.~Alman, B.~Li, Personal communication (2026).

\bibitem{ALPHAEVOLVE2025}
A.~Novikov, N.~Vũ, M.~Eisenberger, E.~Dupont, P.-S. Huang, A.~Z. Wagner,
  S.~Shirobokov, B.~Kozlovskii, F.~J.~R. Ruiz, A.~Mehrabian, M.~P. Kumar,
  A.~See, S.~Chaudhuri, G.~Holland, A.~Davies, S.~Nowozin, P.~Kohli, M.~Balog,
  \href{https://arxiv.org/abs/2506.13131}{Alphaevolve: A coding agent for
  scientific and algorithmic discovery} (2025).
\newblock \href {http://arxiv.org/abs/2506.13131} {\path{arXiv:2506.13131}}.
\newline\urlprefix\url{https://arxiv.org/abs/2506.13131}

\bibitem{DUMAS2025}
J.-G. Dumas, C.~Pernet, A.~Sedoglavic,
  \href{https://arxiv.org/abs/2506.13242}{A non-commutative algorithm for
  multiplying 4x4 matrices using 48 non-complex multiplications} (2025).
\newblock \href {http://arxiv.org/abs/2506.13242} {\path{arXiv:2506.13242}}.
\newline\urlprefix\url{https://arxiv.org/abs/2506.13242}

\bibitem{SEDOGLAVICTABLE}
A.~Sedoglavic, \href{https://fmm.univ-lille.fr/}{Collection of fast matrix
  multiplication algorithms}, accessed: 2026-02-03 (2025).
\newline\urlprefix\url{https://fmm.univ-lille.fr/}

\bibitem{DEGROOTE19781}
H.~F. {de Groote},
  \href{https://www.sciencedirect.com/science/article/pii/0304397578900385}{On
  varieties of optimal algorithms for the computation of bilinear mappings i.
  the isotropy group of a bilinear mapping}, Theoretical Computer Science 7~(1)
  (1978) 1--24.
\newblock \href {https://doi.org/https://doi.org/10.1016/0304-3975(78)90038-5}
  {\path{doi:https://doi.org/10.1016/0304-3975(78)90038-5}}.
\newline\urlprefix\url{https://www.sciencedirect.com/science/article/pii/0304397578900385}

\bibitem{DEGROOTE19782}
H.~F. {de Groote},
  \href{https://www.sciencedirect.com/science/article/pii/0304397578900452}{On
  varieties of optimal algorithms for the computation of bilinear mappings ii.
  optimal algorithms for 2 × 2-matrix multiplication}, Theoretical Computer
  Science 7~(2) (1978) 127--148.
\newblock \href {https://doi.org/https://doi.org/10.1016/0304-3975(78)90045-2}
  {\path{doi:https://doi.org/10.1016/0304-3975(78)90045-2}}.
\newline\urlprefix\url{https://www.sciencedirect.com/science/article/pii/0304397578900452}

\bibitem{VONZURGATHEN1999}
J.~V.~Z. Gathen, J.~Gerhard, Modern Computer Algebra, 2nd Edition, Cambridge
  University Press, USA, 2003.

\bibitem{KEMPER2025}
A.~Kemper, \href{https://github.com/a1880/matrix-multiplication}{From ${F}_2$
  to $\mathbb{Z}$ solutions of {Brent Equations}}, preprint. Available at
  \url{https://github.com/a1880/matrix-multiplication} (July 2025).
\newline\urlprefix\url{https://github.com/a1880/matrix-multiplication}

\end{thebibliography}

\end{document}